%% file: sttforall-lfmtp.tex
\documentclass[usenames,dvipsnames,submission,copyright,creativecommons]{eptcs}
\providecommand{\event}{LFMTP 2018}
\usepackage{underscore}

\input{prelude.tex}

\title{Sharing a Library between Proof Assistants: Reaching out to the HOL Family\footnote{Associated webpage at \url{http://www.lsv.fr/\~fthire/research/sttforall/index.php}.}}
\author{François Thiré
  \institute{ENS Paris-Saclay, LSV,  CNRS, Université Paris-Saclay, INRIA, LIX}
  \email{\quad francois.thire@lsv.fr}
}
\def\titlerunning{Sharing a Library between Proof Assistants: Reaching out to the HOL Family}
\def\authorrunning{François Thiré}
\begin{document}
\maketitle

\begin{abstract}
We observe today a large diversity of proof systems. This diversity  has the negative consequence that a lot of theorems are proved many times. Unlike programming languages, it is difficult for these systems to co-operate because they do not implement the same logic. Logical frameworks are a class of theorem provers that overcome this issue by their capacity of implementing various logics. In this work, we study the \sttabd{} logic, an extension of Simple Type Theory that has been encoded in the logical framework Dedukti~\cite{DBLP:conf/tlca/CousineauD07}. We present a translation from this logic to OpenTheory~\cite{hurd2011}, a proof system and interoperability tool between provers of the HOL family. We have used this translation to export an arithmetic library containing Fermat's little theorem to OpenTheory and to two other proof systems that are Coq~\cite{thecoqdevelopmentteam20171133970} and Matita~\cite{DBLP:conf/cade/AspertiRCT11}.
\end{abstract}

\section{Introduction}
Since Automath and LCF, many proof systems have been designed and used to develop computer-checked mathematics, for example Matita, Coq, HOL4, HOL Light, Isabelle/HOL... These systems sometimes implement different logics. In each of them, proving Fermat's little theorem requires proving about \(300\) lemmas which contribute to the arithmetic library of the system. Developing such a library can be tiresome and we may want, instead of recreating it again and again in different systems, to translate this library from one system to another.

The aim of this paper is to present the logic \sttabd{}, an extension of Simple Type Theory that is powerful enough to express easily arithmetic theorems, but weak enough so that it is easy to export theorems from this logic to several other systems, making this logic suitable for interoperability. \sttabd{} has been implemented in the logical framework Dedukti~\cite{DBLP:conf/tlca/CousineauD07}. In order to illustrate its adequacy for exporting theorems, we have successfully implemented a translation from \sttabd{} to Coq and Matita, and in order to target proof systems based on HOL (Higher-Order Logic), we have also implemented a translation from \sttabd{} to OpenTheory~\cite{hurd2011}, which is a proof system for interoperability between the provers of the HOL family. Then, we applied our translations on an arithmetic library that was available in \sttabd{}. The translation from \sttabd{} to OpenTheory is interesting because even if these two systems are very close, they are based on different design choices that make the translation harder than expected. The description of \sttabd{} and its translation to OpenTheory are the two main contributions of this paper.

\subsection{Logical Frameworks and Interoperability}

Sharing proofs between systems is not always possible since some logics are more expressive than others. For example, one may quantify over proofs in the Calculus of Constructions but not in Higher-Order Logic. Moreover, it is not conceivable to develop, for every pair of proof systems, a specific translation because there would be a quadratic number of translations.

Logical frameworks offer an approach to overcoming these two issues. They are a special kind of proof systems in which different logics and proof systems can be specified. Using a logical framework, there are two ways of sharing proofs. Suppose that the proof of a theorem \(thm_A\) expressed in a logic \(\mathcal{L}_A\) needs a proof of a theorem \(thm_B\) already proven in the logic \(\mathcal{L}_B\). A first solution explored, for instance, by Cauderlier and Dubois~\cite{DBLP:conf/itp/CauderlierD17} is to have the combined proof inside the logical framework by encoding the proofs of \(thm_A\) and \(thm_B\) in it. In this solution, proofs are not exported outside of the logical framework.

Another solution is to translate \(thm_A\) directly from \(\mathcal{L}_A\) to \(\mathcal{L}_B\). This process can be decomposed in three steps. The first step translates \(thm_A\) from \(\mathcal{L}_A\) to the encoding of \(\mathcal{L}_A\) in the logical framework \(U\) denoted by \(U[\mathcal{L}_A]\). The second step translates \(thm_A\) from \(U[\mathcal{L}_A]\) to \(U[\mathcal{L}_B]\). Finally, the third step translates \(thm_A\) from \(U[\mathcal{L}_B]\) to the proof system \(\mathcal{L}_B\). While the first and last steps are total functions, the second step is, in general, a partial function since the translation is not always possible: For example, proof irrelevance is not expressible in HOL but it is in the Calculus of Constructions.

This is the solution we used to import an arithmetic library in Dedukti[\sttabd{}]. Originally, this library has been implemented in the proof system Matita. The translation process that goes from Dedukti[Matita] to Dedukti[\sttabd{}] is described in a separate paper~\cite{interop}.

\subsection{How interoperability should work}
\label{sec:interop}

The definition of a function or a type might differ between several proof systems. For example, in Coq, inductive types such as \(\mathbb{N}\) are primitive while in HOL they are encoded. This is a problem for a generic translation:
\begin{itemize}
\item The definition of a type is not unique. For example, in Coq, natural numbers have at least three different definitions. Which one should be used?
\item If one uses an intermediate system to translate these proofs, the encoding from the first system to the intermediate may \textit{degrade} the original definition. Recreating the original type or function definition might be difficult. For example, the translation of the inductive type \(\mathbb{N}\) is translated in Dedukti as four declarations (\(\mathbb{N}\), \(0\), \(S\) and the recursor) with two rewrite rules. But from these declarations and the rewrite rules, it is difficult to identify the definition of an inductive type.
\end{itemize}

This implies that types and functions in general will be axiomatized during the transformations. For example, the arithmetic library we export into OpenTheory comes with \(40\) constants and \(80\) axioms to define. Fortunately, all of these axioms can be proven easily. Among these axioms, one can find
\[ \forall x, \forall y, x = y \Rightarrow y = x\]
or
\[ \forall n, \forall P, P~n \Rightarrow (\forall m, (n \leq m \Rightarrow P~m \Rightarrow P~(S~m))) \Rightarrow \forall y, (n\leq y \Rightarrow P~y)\]

Thus, the users of the library have to instantiate \textbf{once} the library with the definitions they want to use. We claim that this is the way interoperability should work because this is very flexible: There is no need to regenerate the whole library if the user wants to change one definition.

\subsection{Contributions}

The contributions of this paper are presented in Figure~\ref{fig:process} and detailed below:
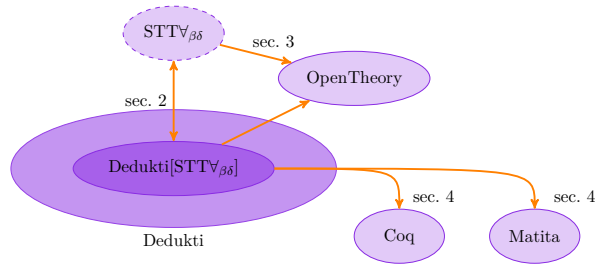
\begin{figure}
  \centering
  \scalebox{0.6}{
    \begin{tikzpicture}
      [node distance=3cm,
      on grid,
      checker/.style = {shape=ellipse, draw=blue-violet!100, fill=blue-violet!25, align=center,
        minimum height=1.2cm, minimum width=2cm},
      checkerb/.style = {dashed, shape=ellipse, draw=blue-violet!100, fill=blue-violet!25, align=center,
        minimum height=1.2cm, minimum width=2cm, label={west:#1}},
      dedukti node/.style = {shape=ellipse, draw=blue-violet!90, fill=blue-violet!50, label={south:Dedukti}},
      post/.style={->, >=stealth', semithick}]

      \node [checker,fill=blue-violet!75] (UA) {Dedukti[\sttabd{}]};
      \begin{scope}[on background layer]
        \node [dedukti node, inner sep=9pt, fit={(UA)}] (Dedukti) {};
      \end{scope}
      \node [checkerb] (STT) [above=of Dedukti] {\sttabd{}};
      \node [checker] (COQ) [yshift=-1.5cm, xshift=2cm, right=of Dedukti] {Coq};
      \node [checker] (MATITA) [right=of COQ] {Matita};
      \node [checker] (OpenTheory) [right=of STT, xshift=1cm, yshift=-1cm] {OpenTheory};
      \draw[post, very thick, <-, orange] (OpenTheory) -- (STT) node[midway, xshift=0.4cm, above] {\textcolor{black}{sec.~\ref{sec:otstt}}};
      \draw[post, very thick, <->, orange] (UA) -- (STT) node[midway,left] {\textcolor{black}{sec.~\ref{sec:background}}};
      \draw[post, very thick, <-, orange] (OpenTheory) -- (UA) ;
      \draw[post, very thick, <-, orange] (COQ)
      .. controls ([yshift=-0cm,xshift=5cm]Dedukti)
      .. (UA)node[midway,right, xshift=0.5cm, yshift=-0.5cm] {\textcolor{black}{sec.~\ref{sec:tocoq}}};
      \draw[post, very thick, <-, orange] (MATITA)
      .. controls ([yshift=1.5cm,xshift=3cm]COQ)
      .. (UA)node[midway,right, xshift=1cm, yshift=-0.5cm] {\textcolor{black}{sec.~\ref{sec:tocoq}}};
    \end{tikzpicture}
  }
    \caption{Contribution}
\label{fig:process}
\end{figure}
\begin{itemize}
\item We introduce a new logic namely \sttabd{} (Section~\ref{sec:background}) as well as its encoding in Dedukti;
\item We give a translation from \sttabd{} to OpenTheory (Section~\ref{sec:otstt});
\item We describe the embedding from \sttabd{} to Coq and Matita (Section~\ref{sec:tocoq});
\item We describe the translation of an arithmetic library from Dedukti[\sttabd{}] to OpenTheory, Coq and Matita (Section~\ref{sec:library}).
\end{itemize}

As we will see in Section~\ref{sec:background}, Dedukti[\sttabd{}] and \sttabd{} define the same logic. To simplify the presentation of the translation of this logic to OpenTheory, we choose to describe it as a translation from \sttabd{} even if the actual implementation is a translation from Dedukti[\sttabd{}].

\section{Dedukti and \sttabd}
\label{sec:background}
\subsection{Dedukti}

Dedukti is a logical framework that implements the \lpc{}~\cite{DBLP:conf/tlca/CousineauD07}~\cite{DBLP:phd/hal/Saillard15a}, a calculus that extends the \(\lambda\Pi\) calculus (also known as LF)~\cite{DBLP:journals/jacm/HarperHP93} with rewrite rules. These rules can be used for the convertibility test. The syntax is the following:

\begin{tabular}{cccc}
  Terms & \(A,B,t,u\) & ::= & \Kind ~|~\Type~|~\(\mathrm{\product{x}{A}{B}}\)~|~\(\app{A}{B}\)~|~\(\absT{x}{A}B{}\)~|~x\\
  Contexts & \(\Gamma \) & ::= & \(\emptyset\) ~|~ \(\Gamma, x:A\) ~|~ \(\Gamma, t \hookrightarrow u\)\\
\end{tabular}

and the type system of \lpc{} are respectively presented in Figure~\ref{lambdapi}. For simplicity, we do not present how to derive the judgment \(\Gamma \vdash t \hookrightarrow u~\mathbf{wf}\) here, but it can be found in Saillard's thesis~\cite{DBLP:phd/hal/Saillard15a}. Roughly, a rewrite rule \(t \hookrightarrow u\) is well typed when the types of \(t\) and \(u\) are convertible. One advantage of using rewrite rules is that more systems can be encoded in Dedukti using a \textit{shallow} encoding where by shallow we mean an encoding having the two following properties: 1) a binder of the source language is translated as a binder in the second language (using HOAS (Higher-Order Abstract Syntax)~\cite{Pfenning:1988:HAS:960116.54010} for example), and 2) the typing judgment in the source language is translated as a typing judgment in Dedukti. This means that we can use the type checker of Dedukti to check directly if a term from the source language encoded in Dedukti is well typed. The next two paragraphs are dedicated to the \sttabd{} system and its shallow encoding in Dedukti.

\begin{rules}{\lp{} modulo theory typing system}{lambdapi}
  \addrule{\Rule{\nil~\wf}{\dkstyle{\dkstyle{dk empty ctx}}}}
  \addrule{\RuleP{\Gamma \vdash A : s}{\Gamma, x:A~\mathbf{wf}}{\dkstyle{dk ctx sort}}}
  \addrule{\RuleP{\Gamma \vdash t\hookrightarrow u ~\mathbf{wf}}{\Gamma, t \hookrightarrow u~\mathbf{wf}}{dk \dkstyle{ctx \(\hookrightarrow\)}}}
  \addrule{\RuleP{\Gamma~\wf}{\Gamma \vdash \Type:\Kind}{\dkstyle{ctx \Type{} \Kind}}}
  \addrule{\RulePP{\Gamma~\wf}{(x:A) \in \Gamma}{ \Gamma \vdash x:A}{\dkstyle{dk var}}}
  \addrule{\RulePP{\Gamma \vdash A:\Type}{ \Gamma, x:A \vdash B : s}{\Gamma \vdash \PiT{x}{A}{B}:s}{\dkstyle{dk \(\Pi\)}}}
  \addrule{\RulePPP{\Gamma \vdash A:s}{ \Gamma, x:A \vdash B : s}{\Gamma,x:A \vdash t:B}{\Gamma \vdash \absT{x}{A}{t}:\PiT{x}{A}{B}}{\dkstyle{dk \(\lambda\)}}}
  \addrule{\RulePP{\Gamma \vdash t:\PiT{x}{A}{B}}{ \Gamma \vdash t' : A}{\Gamma \vdash t~t':B[x \leftarrow t']}{\dkstyle{dk app}}}
  \addrule{\RulePPP{\Gamma \vdash A:s}{ \Gamma \vdash t : A}{A\equiv_{\beta\Gamma} B}{\Gamma \vdash t:B}{\dkstyle{dk \(\equiv\)}}}
\end{rules}
\subsection{\sttabd}
\label{sec:sttm}

\sttabd{} is an extension of Simple Type Theory with prenex polymorphism and type operators.  A type operator is constructed using a name and an arity. This allows to declare types such as \texttt{bool}, \texttt{nat} or \texttt{\(\alpha\) list}. The polymorphism of \sttabd{} is restricted to prenex polymorphism as full polymorphism would make this logic inconsistent\footnote{Coquand's paper~\cite{DBLP:conf/lics/Coquand86} shows also that omitting types annotations for polymorphic types would make the logic inconsistent}~\cite{DBLP:conf/lics/Coquand86}.  The \sttabd{} syntax is presented in Fig.~\ref{fig:sttsyntax}. The type of propositions \(\Prop\) and the type of functions \(\to\), could be declared as type operators, of arity \(0\) and \(2\) respectively. Since they have a particular meaning for the typing judgment, we add them explicitly. Also, \sttabd{} allows to declare and define constants. Declaring constants is better for interoperability as discussed in section~\ref{sec:interop} but increases the number of axioms that need to be ultimately instantiated. The typing system and the proof system are presented in Fig.~\ref{fig:stttyping} and Fig.~\ref{fig:sttproof}. Finally, we point out that we identify in \sttabd{} the terms \(t\) and \(t'\) if they are convertible up to \(\beta\) and \(\delta\) (unfolding of constants).

\begin{figure}
  \centering
  \begin{tabular}{lccl}
    \textbf{Type operator} & \(p\) & & \\
    \textbf{Type variable} & \(X\) & & \\
    \textbf{Monotype} & \(A,B\) & \(\defn\) & \(X\)~|~\(\Prop\)~|~\(A \to B\)~|~\(p~A_1~...~A_n\)\\
    \textbf{Polytype} & \(T\) & \(\defn\) & \(A\)~|~\(\forallT{X}{}{T}\)\\
    \hline
    \textbf{Constant} & \(cst\) & & \\
    \textbf{Constant term} & \(c\) & \(\defn\) & \(cst\)~|~\(c~A\)\\
    \hline
    \textbf{Term variable} & \(x\) & & \\
    \textbf{Monoterm} & \(t,u\) & \(\defn\) & \(x\)~|~\(c\)~|~\(\absT{x}{A}{t}\) ~|~ \(\app{t}{u}\) ~|~\(t \Rightarrow u\)~|~ \(\forallT{x}{A}{t}\) ~|~ \(\absT{X}{}{t}\)\\
    \textbf{Polyterm} & \(\tau\) & \( \defn\) & \(t\)~|~\(\forallP{X}{}{\tau}\)\\
    \hline
    \textbf{Typing Context} & \(\Gamma\) & \(\defn\) & \(\emptyset\)~|~\(\Gamma, t:T\)~|~\(\Gamma, X\)\\
    \textbf{Proof Context} & \(\Xi \) & \(\defn\) & \(\emptyset\)~|~\(\Xi,t\)\\
    \textbf{Constant Context} & \(\Sigma\) & \(\defn\) & \(\emptyset\)~|~\(\Sigma, cst = \tau : T\)~|~\(\Sigma, cst : T\)~|~\(\Sigma, (p:n)\)\\
    \hline
    \textbf{Typing Judgment} & \(\mathcal{T}\) & \(\defn\) & \(\Sigma;\Gamma \vdash \tau : T\)\\
    \textbf{Proof Judgment} & \(\mathcal{P}\) & \(\defn\) & \(\Sigma;\Gamma;\Xi \vdash \tau\)\\
    \hline
    \textbf{MonoType well-formed} & & & \( \Sigma;\Gamma \vdash A~\mathbf{wf}\)\\
    \textbf{PolyType wellf-formed} & & & \( \Sigma;\Gamma \vdash T~\mathbf{wf}\)\\
    \textbf{Typing ctx well-formed} & & & \(\Sigma \vdash \Gamma~\mathbf{wf}\) \\
    \textbf{Constant ctx well-formed} & & & \(\Sigma~\mathbf{wf}\) \\
  \end{tabular}
\caption{\sttabd{} syntax}
\label{fig:sttsyntax}
\end{figure}

\begin{rules}{\sttabd{} typing system}{fig:stttyping}
  \addrule{\RuleP{X \in \Gamma}{\Sigma;\Gamma \vdash X~\mathbf{wf}}{\sttstyle{S WF var}}}
  \addrule{\Rule{\Sigma;\Gamma \vdash \Prop~\mathbf{wf}}{\sttstyle{S WF prop}}}
  \addrule{\RulePP{\Sigma;\Gamma \vdash A~\mathbf{wf}}{\Sigma;\Gamma \vdash B~\mathbf{wf}}{\Sigma;\Gamma \vdash A\to B~\mathbf{wf}}{\sttstyle{S WF fun}}}
  \addrule{\RuleP{\Sigma \vdash \Gamma~\mathbf{wf}}{\Sigma \vdash \Gamma,X~\mathbf{wf}}{\sttstyle{S WF ctx var}}}
  \addrule{\RulePP{(p:n) \in \Sigma}{\Sigma;\Gamma \vdash A_i~\mathbf{wf}}{\Sigma;\Gamma \vdash p~A_1~\dots~A_n~\mathbf{wf}}{\sttstyle{S WF tyop app}}}
  \addrule{\Rule{\emptyset~\mathbf{wf}}{\sttstyle{S WF empty}}}
  \addrule{\RulePP{\Sigma \vdash \Gamma~\mathbf{wf}}{\Sigma;\Gamma \vdash A~\mathbf{wf}}{\Sigma \vdash \Gamma,x:A~\mathbf{wf}}{\sttstyle{S WF ctx var}}}
  \addrule{\RulePP{\Sigma~\mathbf{wf}}{p \not\in Dom(\Sigma)}{\Sigma,(p:n)~\mathbf{wf}}{\sttstyle{S WF tyop}}}
  \addrule{\RuleP{\Sigma;\Gamma,X \vdash T~\mathbf{wf}}{\Sigma;\Gamma \vdash \forallT{X}{}{T}~\mathbf{wf}}{\sttstyle{S WF forall ty}}}
  \addrule{\RulePPP{\Sigma~\mathbf{wf}}{cst \not\in Dom(\Sigma)}{\Sigma \vdash T~\mathbf{wf}}{\Sigma, cst : T~\mathbf{wf}}{\sttstyle{S WF cst decl}}}
  \addrule{\RulePPP{\Sigma~\mathbf{wf}}{cst \not\in Dom(\Sigma)}{\Sigma;\emptyset;\emptyset \vdash \tau : T}{\Sigma, cst = \tau : T~\mathbf{wf}}{\sttstyle{S WF cst defn}}}
  \addrule{\RulePP{\Sigma~\mathbf{wf}}{\Sigma \vdash \Gamma~\mathbf{wf}}{\mathcal{C}, x:A \vdash x:A}{\sttstyle{S var}}}
  \addrule{\RulePP{ \mathcal{C} \vdash f : A \to B}{ \mathcal{C} \vdash t : A}{\mathcal{C} \vdash \app{f}{t} : B}{\sttstyle{S app}}}
  \addrule{\RuleP{ \mathcal{C}, x:A \vdash t:B}{ \mathcal{C} \vdash \absT{x}{A}{t} : A \to B}{\sttstyle{S abs}}}
  \addrule{\RulePP{ \mathcal{C} \vdash t : \Prop}{ \mathcal{C} \vdash u : \Prop}{ \mathcal{C} \vdash t \Rightarrow u : \Prop}{\sttstyle{S imp}}}
  \addrule{\RuleP{ \mathcal{C}, x:A \vdash t: \Prop}{ \mathcal{C} \vdash \forallT{x}{A}{t} : \Prop}{\sttstyle{S forall}}}
  \addrule{\RuleP{ \mathcal{C},X \vdash t:T}{ \mathcal{C} \vdash \abs{X}{t} : \forallT{X}{}{T}}{\sttstyle{S poly intro}}}
  \addrule{\RulePP{FV(B) \subseteq \mathcal{C}}{\mathcal{C} \vdash c : \forallP{X}{T}}{\mathcal{C} \vdash c~B : \subst{T}{X}{B}}{\sttstyle{S cst app}}}
  \addrule{\RuleP{typeof(\mathcal{C},cst) = T}{\mathcal{C} \vdash cst : T}{\sttstyle{S cst}}}
  \addrule{\RuleP{\mathcal{C}, X \vdash \tau : \Prop}{\mathcal{C} \vdash \forallP{X}{\tau} : \Prop}{\sttstyle{S Poly}}}
\end{rules}

\begin{rules}{\sttabd{} proof system}{fig:sttproof}
  \addrule {\RuleP{ \mathcal{C} \vdash t : \Prop}{\mathcal{C},t \vdash t}{\sttstyle{S assume}}}
  \\
  \addrule {\RulePP{ \mathcal{C} \vdash t }{\mathcal{C} \vdash t \Rightarrow u}{ \mathcal{C} \vdash u }{\sttstyle{S \(\Rightarrow_E\)}}}

  \addrule{\RuleP{ \mathcal{C}, t\vdash u }{ \mathcal{C} \vdash t \Rightarrow u }{\sttstyle{S \(\Rightarrow_I\)}}}

  \addrule{\RulePP{ \mathcal{C}\vdash \forallT{x}{A}{t} }{ \mathcal{C} \vdash u : A}{ \mathcal{C} \vdash \subst{t}{x}{u} }{\sttstyle{S \(\forall_E\)}}}

  \addrule{\RulePP{ \mathcal{C},x:A \vdash t }{x \not\in \mathcal{C}}{ \mathcal{C} \vdash \forallT{x}{A}{t} }{\sttstyle{S \(\forall_I\)}}}
  \addrule {\RulePP{ \mathcal{C} \vdash \forallP{X}{\tau}}{FV(A)\subseteq \Gamma}{\mathcal{C} \vdash \tau[X:=A]}{\sttstyle{S \(\forallp_E\)}}}
  \addrule {\RuleP{ \mathcal{C},X\vdash \tau}{\mathcal{C} \vdash \forallP{X}{\tau}}{\sttstyle{S \(\forallp_I\)}}}
  \addrule{\RulePP{\mathcal{C}\vdash t}{t \equiv_{\beta\delta} t'}{\mathcal{C} \vdash t'}{\sttstyle{S conv}}}
\end{rules}

\sttabd{} is sound and type checking is decidable. The proofs are provided in annex of this paper.

\begin{theorem}
  \label{thm:sttsound}
  \sttabd{} is sound: the judgment \(\emptyset \vdash \forallT{x}{\Prop}{x}\) is not provable.
\end{theorem}
\begin{proof}
  We construct a set-theoretical model of \sttabd{}.
\end{proof}

\begin{theorem}
  \label{thm:sttdecidable}
  Type checking and proof checking in \sttabd{} are decidable.
\end{theorem}
\begin{proof}
  We only have to show that \(\equiv_{\beta\delta}\) is decidable. This property follows from the fact that \(\hookrightarrow_{\beta} \cup \hookrightarrow_{\delta}\) is a convergent term rewriting system.
\end{proof}
\subsubsection{Equality in \sttabd{}}

\begin{wrapfigure}{r}{0.5\textwidth}
\scalebox{0.8}{
\begin{mathpar}
  \inferrule* [Right=S \(\forallp_I\)]
  {
    \inferrule*[Right=S \(\forall_I\)]
    {
      \inferrule*[Right=S conv]
      {
        \inferrule*[Right=S \(\forall_I\)]
        {
          \inferrule*[Right=S \(\Rightarrow_I\)]
          {
            \inferrule*[Right=S assume]
            { }
            {=_{\mathcal{L}}; X, x:X, p:X \to \Prop; P~x \vdash P~x}
          }
          {=_{\mathcal{L}}; X, x:X, p:X \to \Prop; \emptyset \vdash P~x \Rightarrow P~x}
        }
        {=_{\mathcal{L}}; X, x:X; \emptyset \vdash \forallT{P}{X \to \Prop}{P~x \Rightarrow P~x}}
      }
      {=_{\mathcal{L}}; X, x:X;\emptyset \vdash x =_{\mathcal{L}} x}
    }
    {=_{\mathcal{L}};X;\emptyset \vdash \forallT{x}{X}{x =_{\mathcal{L}} x}}
  }
  { =_{\mathcal{L}} ;\emptyset ; \emptyset \vdash \forallP{X}{}{\forallT{x}{X}{x =_{\mathcal{L}} x}}}
\end{mathpar}
}
\end{wrapfigure}
In order to give further insights into \sttabd{}, we give here an example that expresses polymorphic Leibniz equality denoted by \(=_{\mathcal{L}}\) in \sttabd{}. Its type will be \(\forallT{X}{}{X \to X \to \Prop}\) and it can be implemented by the term \(\absT{X}{}{\absT{x}{X}{\absT{y}{X}{\forallT{P}{X \to \Prop}{P~x \Rightarrow P~y}}}}\). From this definition, it is possible to prove that \(=_{\mathcal{L}}\) is reflexive, which is expressed by the statement \( \forallP{X}{}{\forallT{x}{X}{x =_{\mathcal{L}} x}} \), proved on the right.

\subsection{Dedukti[\sttabd]}
\label{sec:deduktistt}

The purpose of this section is to describe a shallow encoding of \sttabd{} in Dedukti. Thanks to HOAS~\cite{Pfenning:1988:HAS:960116.54010}, such a shallow encoding exists in Dedukti. In Figure~\ref{fig:sigtypes}, we present the signature used to encode terms from \sttabd{}. The translation function is given in annexes.

For the types, we declare in Dedukti two symbols \(type\) and \(ptype\) that are used to encode monomorphic types and polymorphic types of \sttabd{}. Therefore every type of \sttabd{} will be encoded as an object of Dedukti. That is why we use the symbol \(term\) to encode a \sttabd{} type to a Dedukti type. We add in the signature a symbol \(p\) to coerce a monomorphic type to a polymorphic type. Then we need symbols to represent type constructors of \sttabd{}: The Dedukti's symbol \(prop\) encodes \(\Prop\) while \(arr\) encodes \(\to\). Each type constructor of arity \(n\) is encoded as a new Dedukti symbol of type \(type \to \cdots \to type\) with \(n+1\) occurrences of \(type\). Finally, to encode \(\forall\) at the type level, we use the Dedukti symbol \(forallK_{type}\).

For the terms, since the encoding is shallow, we do not need symbols for abstractions and applications. In contrast, we need the two symbols \(forall\) and \(impl\) that encode respectively the connectives \(\forall\) and \(\Rightarrow\). Then, we add the symbol \(forallK_{prop}\) to encode polymorphic propositions. To encode a proposition into a Dedukti type, we use the symbol \(proof\). Finally, rewrite rules transform a deep representation of \sttabd{} syntax to a shallow one, for instance the Dedukti rule \(term~(p~(arr~l~r)) \hookrightarrow term~l \to term~r\) allows the Dedukti term \(term~(p~(arr~l~r))\) to be the type of a Dedukti's abstraction.

\begin{figure}
  \centering
  \begin{equation*}
  \begin{aligned}
    \mathsf{type}  &: \mathbf{Type} \\
    \mathsf{arr}   &: \mathsf{type \to type \to type}\\
    \mathsf{prop}  &: \mathsf{type}\\
    \mathsf{ptype} &: \mathbf{Type} \\
    \mathsf{p}     &: \mathsf{type \to ptype}\\
    \mathsf{term}  &: \mathsf{ptype \to \mathbf{Type}}
  \end{aligned}
  \qquad
  \begin{aligned}[c]
    \mathsf{impl}  &: \mathsf{term~(p~(arr~prop~(arr~prop~prop))}\\
    \mathsf{forallK_{type}} &: \mathsf{(type \to ptype) \to ptype}\\
    \mathsf{proof} &:  \mathsf{term~(p~prop) \to \mathtt{Type}}\\
    \mathsf{forall} &:  \mathsf{(t:~type) \to~term~(p~(arr~(arr~t~prop)~prop))}\\
    \mathsf{forallK_{prop}} &: \mathsf{(type \to term~(p~prop)) \to term~(p~prop)}
  \end{aligned}
\end{equation*}
\begin{equation*}
    \begin{aligned}[c]
    \mathsf{term~(p (arr~l~r))}&\rw \mathsf{term~(p~l) \to term~(p~r)}\\
    \mathsf{term~(forallK_{type}~f)} &\rw \mathsf{(x:type) \to term~(f~x)}\\
    \mathsf{proof~(forall~t~f)} &\rw \mathsf{(x: term~(p~t)) \to proof~(f~x)} \\
    \mathsf{proof~(impl~l~r)} &\rw \mathsf{proof~l \to proof~r}\\
    \mathsf{proof~(forallK_{prop}~f)} &\rw \mathsf{(x:type) \to proof~(f~x)}
  \end{aligned}
\end{equation*}
  \caption{Signature for \sttabd{} in Dedukti}
  \label{fig:sigtypes}
\end{figure}

\subsubsection{A proof of reflexivity in Dedukti[\sttabd{}]:}

The translation of Leibniz equality in Dedukti[\sttabd] is as follow. First, the type of \(=_{\mathcal{L}}\)\footnote{also written \(leibniz\) in its prenex form} is translated as:
\[term~(forallK_{type}~(\absT{X}{}{arr~X~(arr~X~prop)})\]
then its definition is translated as
\[\absT{A}{}{\absT{x}{term~A}{\absT{y}{term~A}{\forallT{P}{term~(arr~A~prop)}{impl~(P~x)~(P~y)}}}}\]
Finally, the proof of \texttt{refl} is translated as
\[\absT{A}{}{\absT{x}{term~A}{\absT{P}{term~(arr~A~prop)}{\absT{h}{proof~(P~x)}{h}}}}\]
that is of type
\[proof~(forallK_{prop}~(\absT{X}{}{\forallT{x}{term X}{leibniz~X~x~x}}))\]
which is the translation of
\[\forallP{X}{}{\forallT{x}{X}{x =_L x}}\]

\section{OpenTheory}
\label{sec:otstt}

HOL is a logic that is implemented in several systems with some minor differences. OpenTheory~\cite{hurd2011} is a tool that allows to share proofs between several implementations of HOL. Since we are targeting OpenTheory, we will mostly refer to the logic defined by OpenTheory. The logic behind OpenTheory comes from \qzero{}~\cite{DBLP:books/daglib/0070479}, a classical logic taking only equality as a primitive logical connective. Terms are those of Simply Typed Lambda Calculus with an equality symbol while the type system extends the one of the Simply Typed Lambda Calculus by declaring type operators and prenex polymorphism. The syntax and the proof system of OpenTheory can be found respectively in Fig.~\ref{fig:asshole} and Fig.~\ref{opentheory}. The syntax being very similar to the one of \sttabd{}, we have omitted the definitions of typing judgments.
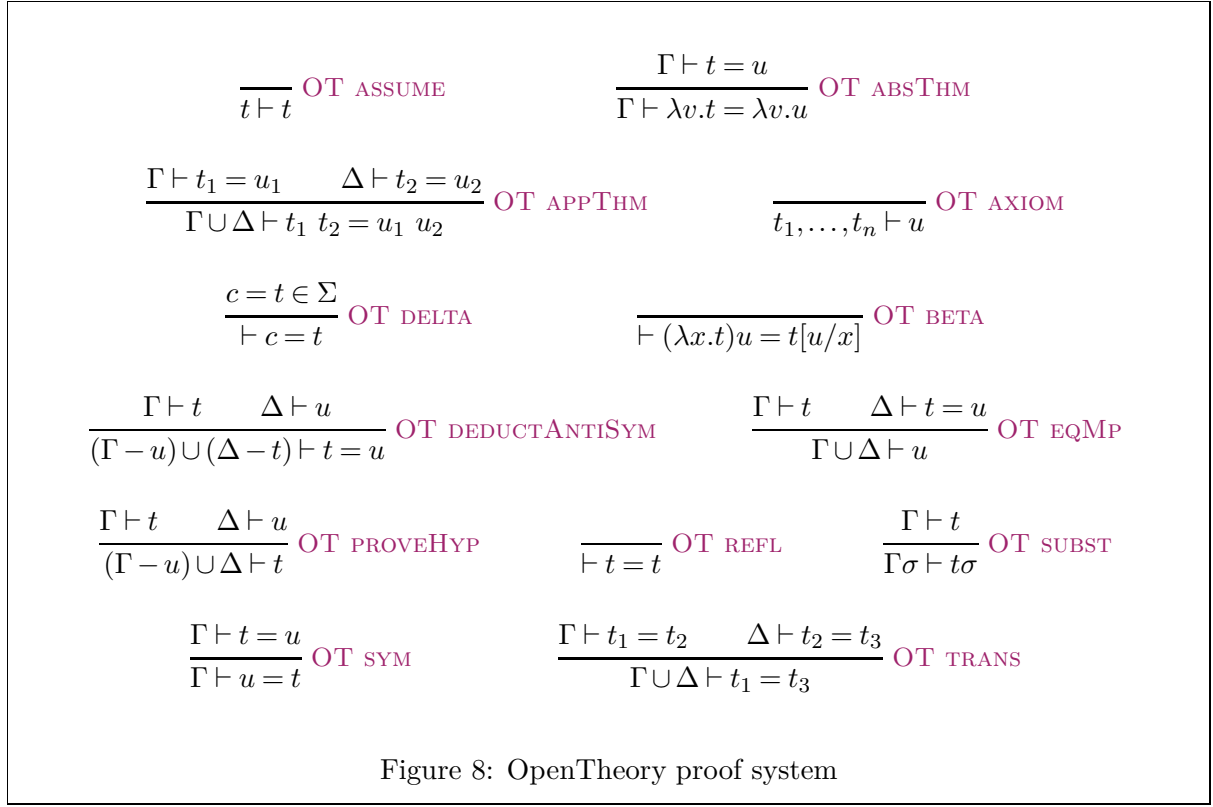
\begin{figure}
  \centering
  \begin{tabular}{lccl}
    \textbf{Type operators} & p & & \\
    \textbf{Type variables} & X & & \\
    \textbf{Types} & \(A,B\) & \(\defn\) & X~|~\(\Prop\)~|~\(A \to B\)~|~\(p~A_1~...~A_n\)\\
    \hline
    \textbf{Terms variables} & x & & \\
    \textbf{Terms} & t,u & \(\defn\) & x~|~c~|~\(\absT{x}{A}{t}\) ~|~ \(\app{t}{u}\) ~|~\(t = u\)~|~ \(c\) \\
    \hline
    \textbf{Typing Judgment} & \(\mathcal{T}\) & \(\defn\) & \(\Sigma;\Gamma \vdash \tau : T\)\\
    \textbf{Proof Judgment} & \(\mathcal{P}\) & \(\defn\) & \(\Sigma;\Gamma;\Xi \vdash \tau\)\\
  \end{tabular}
  \caption{OpenTheory syntax}
  \label{fig:asshole}
\end{figure}

\begin{rules}{OpenTheory proof system}{opentheory}
  \addrule{\Rule{t \vdash t}{\otstyle{OT assume}}}
  \addrule{\RuleP{\Gamma \vdash t = u}{\Gamma \vdash \lambda v.t = \lambda v.u}{\otstyle{OT absThm}}}
  \addrule{\RulePP{\Gamma \vdash t_1 = u_1}{\Delta \vdash t_2 = u_2}{\Gamma \cup \Delta \vdash t_1~t_2 = u_1~u_2}{\otstyle{OT appThm}}}
  \addrule{\Rule{t_1, \dots, t_n \vdash u}{\otstyle{OT axiom}}}
  \addrule{\RuleP{c = t \in \Sigma}{\vdash c = t}{\otstyle{OT delta}}}
  \addrule{\Rule{\vdash (\lambda x.t) u = t[u/x]}{\otstyle{OT beta}}}
  \addrule{\RulePP{\Gamma \vdash t}{\Delta \vdash u}{(\Gamma - u) \cup (\Delta - t) \vdash t = u}{\otstyle{OT deductAntiSym}}}
  \addrule{\RulePP{\Gamma \vdash t }{\Delta \vdash t = u}{\Gamma \cup \Delta \vdash u}{\otstyle{OT eqMp}}}
  \addrule{\RulePP{\Gamma \vdash t}{\Delta \vdash u}{(\Gamma - u) \cup \Delta \vdash t}{\otstyle{OT proveHyp}}}
  \addrule{\Rule{\vdash t = t}{\otstyle{OT refl}}}
  \addrule{\RuleP{\Gamma \vdash t}{\Gamma\sigma \vdash t\sigma}{\otstyle{OT subst}}}
  \addrule{\RuleP{\Gamma \vdash t = u}{\Gamma \vdash u = t}{\otstyle{OT sym}}}
  \addrule{\RulePP{\Gamma \vdash t_1 = t_2}{\Delta \vdash t_2 = t_3}{\Gamma \cup \Delta \vdash t_1 = t_3}{\otstyle{OT trans}}}
\end{rules}

\subsection{\sttabd{} vs OpenTheory}

One may notice that \sttabd{} and \ot{} are quite similar. However, there are some differences that makes the translation from \sttabd{} and \ot{} not so easy:

\begin{itemize}
\item Terms in \ot{} are only convertible up to \(\alpha\) conversion while in \sttabd{} it is up to \(\alpha,\beta,\delta\) conversion
\item All the connectives of \ot{} are defined from the equality symbol, while in \sttabd{} they are defined from \(\forall\) and \(\Rightarrow\) connectives
\item Prenex polymorphism in \ot{} is implicit: All free type variables in \ot{} are implicitly quantified while in \sttabd{} all quantifications are explicit
\end{itemize}
These differences lead to three different proof transformations:
\begin{itemize}
\item Encode the \(\forall\) and \(\Rightarrow\) connectives using the equality of OpenTheory
\item Explicit each application of the conversion rule
\item Finally, get rid of the type quantifier
\end{itemize}
OpenTheory is a classical logic while \sttabd{} is intuitionistic. This is not an issue here since intuitionistic logic is a fragment of classical logic.

\subsection{From \sttabd{} to OpenTheory}
\subsubsection{Encoding \(\forall\) and \(\Rightarrow\) using the equality}
\label{sec:forallimpl}
A first idea to encode \sttabd{} proofs in OpenTheory would be to axiomatize all the rules of \sttabd{} and then translate the proofs using these axioms. But translating a rule to an axiom in OpenTheory requires the use of implication. Since \ot{} does not know what an implication is, such axioms would not be usable since it would not be possible to use the modus ponens to eliminate the implication itself. Therefore, one must find an encoding of the \(\forall\) and \(\Rightarrow\) connectives such that the rules of \sttabd{} are admissible. Such encoding is already known from \(\mathcal{Q}_0\)~\cite{DBLP:books/daglib/0070479}. This encoding is presented below and uses two other connectives that are \(\top\) and \(\land\) that can be defined as axiom in OpenTheory:

\begin{align*}
    \top & = \abs{x}{x} & x \Rightarrow y &= (x \land y) = x \\
    x \land y &= \abs{f}{\app{\app{f}{x}}{y}} = \abs{f}{\app{\app{f}{\top}}{\top}} &  \forall x. P& = \absT{x}{}{P} = \abs{x}{\top}
\end{align*}

We stress here that it is really important to axiomatize these definitions and not to define new constants. The difference is that it will be possible to instantiate later these connectives by the \textit{true} connectives of HOL as long as these axioms can be proved regardless of their definition in HOL. These axioms are not too strong to satisfy because in HOL, extensionality of predicates\footnote{\(\forall P, \forall Q, (\forall x, P~x = Q~x \Rightarrow P = Q)\)} is admissible. Using this encoding, it is possible to derive all the rules of \sttabd{} in \ot{} using the four axioms above. Below, we prove the admissibility of the \textsc{S \(\forall_{I}\)} rule
\begin{mathpar}
  \inferrule*[Right=S \(\forall_{I}\)]
  {
    \inferrule*
    {
      \Pi
    }
    {\mathcal{C},x:A \vdash t }
    \\
    \inferrule*
    {}
    {x \not\in \mathcal{C}}
  }
  {\mathcal{C} \vdash \forallT{x}{A}{t} }
\end{mathpar}
using the derivation tree below, \(\Gamma\) is the translation of \(\mathcal{C}\) in OpenTheory\footnote{In OpenTheory, free variables such as \(x\) do not need to appear inside the context.}.

  \begin{mathpar}
  \inferrule* [sep=8pt]
  {
    \inferrule*
    {
      \inferrule*
      {
        \inferrule*
        { \Pi }
        { \Gamma \vdash t }
        \\
        \inferrule*
        { }
        { \Gamma \vdash \top }
      }
      {\Gamma \vdash t = \top}
    }
    {\Gamma \vdash \abs{x}{t} = \abs{x}{\top}}
    \\
    \inferrule*
    {
      \inferrule*
      { }
      {\Gamma \vdash \forallT{x}{}{t} = (\abs{x}{t} = \abs{x}{\top})}
    }
    {\Gamma \vdash (\abs{x}{t} = \abs{x}{\top}) = \forallT{x}{}{t}}
  }
  {\Gamma \vdash \forallT{x}{}{t}}
\end{mathpar}

The right branch is closed thanks to the axiom defining \(\forall\).

All the rules of \sttabd{} can be derived in a similar way. At the end of this translation, the syntax of the term is changed: \(=\) becomes a new connective, while \(\forall\) and \(\Rightarrow\) become defined constants.
\subsubsection{Eliminate \(\beta,\delta\) reductions}

In \sttabd{}, the terms \(\forallP{X}{\forallT{x}{X}{x =_{\mathcal{L}} x}}\) and \(\forallP{X}{\forallT{x}{X}{\forallT{P}{}{P~x \rightarrow P~x}}}\) are convertible, but not in OpenTheory. The convertibility test in \sttabd{} will unfold the definition of \(=_{\mathcal{L}}\) once, then it will apply twice a \(\beta\)-reduction. However, in OpenTheory, it is possible to prove
\[\left(\forallP{X}{\forallT{x}{X}{x =_{\mathcal{L}} x}}\right) = \left(\forallP{X}{\forallT{x}{X}{\forallT{P}{}{P~x \rightarrow P~x}}}\right)\]
The purpose of this section is to explain how it is possible to derive a proof of \(t = t'\) in OpenTheory when \(t \equiv_{\beta\delta} t'\) in \sttabd{}. The decidability of type checking in \sttabd{} relies on the decidability of the conversion rule~\textsc{S Conv}. Since the term rewriting system defined by \(\hookrightarrow_{\beta}\) and \(\hookrightarrow_{\delta}\) is convergent, we can decide whether \(t \equiv_{\beta\delta} u\) by computing their normal forms \(t'\) and \(u'\), then checking they are equal up to \(\alpha\)-conversion. OpenTheory has two rules to handle \(\beta\) and \(\delta\) conversion:
\begin{mathpar}
  \inferrule*[Right=delta]
  { }
  { \Gamma \vdash c = t}
  \and
  \inferrule*[Right=beta]
  { }
  { \Gamma \vdash \app{\abs{x}{t}}{u} = \subst{t}{x}{u}}
\end{mathpar}
Hence, one rewrite step will be translated as an equality. The same is true for a sequence of rewrite steps thanks to transitivity of equality. Therefore, the main difficulty is to show how to derive the OpenTheory judgment \(t = u\) from the \sttabd{} judgment \(t \hookrightarrow_{\beta\delta}^* u\).

In general, the \textsc{OT beta} and \textsc{OT delta} rules will be applied inside a term. Thus, we need to show that for any context \(C\), the rule below is admissible:
\begin{mathpar}
  \inferrule*[Right=ctxrule]
  { \Gamma \vdash t \hookrightarrow_{\beta\delta} u }
  { \Gamma \vdash C[t] \hookrightarrow_{\beta\delta} C[u] }
\end{mathpar}
\noindent the base case being either the rule \textsc{OT beta} or \textsc{OT delta}. In our setting, contexts can be defined by the following grammar:

\[ C ::= \cdot~|~ C~u ~|~ t~C ~|~ \abs{x}{C} ~|~ \forallP{X}{}{C} ~|~ C \Rightarrow u ~|~ t \Rightarrow C ~|~ \forallT{x}{A}{C} \]

Notice that our definition of contexts does not depend on the previous translation. However, to prove the admissibility of the rule \textsc{ctxrule} for the \(\Rightarrow\) case for example, we will need to use its definition from equality.

\begin{theorem}
  For every context \(C\), the rule \textsc{ctxrule} is admissible.
\end{theorem}
\begin{proof}
  This is done inductively on the structure of \(C\). There are already two contextual rules in OpenTheory for equality  to handle abstractions and applications. We need to derive the other contextual rules to handle \sttabd{} connectives that are: \(\forallp, \Rightarrow\) and \(\forall\).
  We show here the admissibility of the contextual rule for \(\Rightarrow\) but the derivations for all the other rules are in annex.
  \begin{mathpar}
    \inferrule*
    {
      \inferrule*
      {}
      {\Gamma \vdash p = p'}
      \\
      \inferrule*
      {}
      {\Gamma \vdash q = q'}
    }
    {\Gamma \vdash p \Rightarrow q = p' \Rightarrow q'}
  \end{mathpar}

  \begin{mathpar}
    \inferrule*
    {
      \inferrule*
      {
        \inferrule*
        {
          \inferrule*
          {
            \inferrule*
            { }
            {\Gamma, p\Rightarrow q \vdash p \Rightarrow q}
            \\
            \inferrule*
            {
              \inferrule*
              { }
              {\Gamma, p' \vdash p'}
              \\
              \inferrule*
              {
                \inferrule*
                {}
                {\Gamma \vdash p = p'}
              }
              {\Gamma \vdash p' = p}
            }
            {\Gamma, p' \vdash p}
          }
          {\Gamma, p\Rightarrow q, p' \vdash q}
          \\
          \inferrule*
          {}
          {\Gamma \vdash q = q'}
        }
        {\Gamma, p \Rightarrow q, p' \vdash q'}
      }
      {\Gamma, p \Rightarrow q \vdash p' \Rightarrow q'}
      \\
      \inferrule*
      {}
      {\vdots}
    }
    {\Gamma \vdash (p \Rightarrow q) = (p' \Rightarrow q')}
  \end{mathpar}
  Half of the proof is omitted here but the derivation tree is symmetric. This rule can be used to solve two context cases. In the case where \(C \Rightarrow t\), we instantiate \(q\) and \(q'\) by \(t\). Hence, the right premise is closed by the OpenTheory rule \textsc{OT refl}. The case \(t \Rightarrow C\) can be instantiated in a symmetric way. All the other cases can be derived in a similar way.
\end{proof}

\subsubsection{Suppressing type quantifiers}
\label{sec:quantifiers}
OpenTheory implicitly quantifies over free types variables while in \sttabd{} this is done explicitly thanks to the \(\forall\) on types. This implies that substitution in \sttabd{} is handled by the system while in OpenTheory, the user has to manage substitution to avoid capturing free type variables. For example, the following type in \sttabd{}
\(\forallT{Y}{}{Y\to X}[X:=Y]\) is equal to \(\forallT{Z}{}{Z\to Y}\) while in OpenTheory, the same type \(Y\to X\) is equal to \(Y\to Y\) using the \textsc{OT subst} rule with the substitution \(X \mapsto Y\). This mechanism forces us to replace each bound variable by a fresh variable each time the bound variable is substituted. In \sttabd{}, there are two rules that are concerned by this: \textsc{S cst app} and \textsc{S \(\forallp_E\)}. Renaming bound variables can be done easily using the \textsc{subst} rule of OpenTheory. For example, the rule \textsc{S \(\forallp_E\)}
\begin{mathpar}
  \RulePP{ \mathcal{C} \vdash \forallP{X}{\tau}}{FV(A)\subseteq \Gamma}{\mathcal{C} \vdash \tau[X:=A]}{S \(\forallp_E\)}
\end{mathpar}
is translated as the OpenTheory proof
\begin{mathpar}
  \inferrule*[Right=subst]
  {
    \inferrule*[Right=subst]
    {\mathcal{C} \vdash \tau \\ Z~fresh}
    {\mathcal{C} \vdash \tau[X:=Z]}
  }
  {\mathcal{C} \vdash \tau[Z:= A]}
\end{mathpar}

The same thing can be done for the rule \textsc{cst app} each time a constant is applied to a type inside the definition of a constant for example. The rule \textsc{S \(\forallp_I\)} is just removed because there is no need to introduce a quantifier anymore.
\section{From Dedukti[\sttabd{}] to Coq and Matita}
\label{sec:tocoq}
Going from \sttabd{} to Coq or Matita is easy since the Calculus of Inductive Constructions with universes can be seen as an extension of \sttabd{}. Only three universes are needed for the translation: the impredicative universe \(Prop\) for \(\Prop\), \(Type_1\) for monotypes and \(Type_2\) for polytypes. The three \textit{forall} constructions of \sttabd{}, the arrow on types and the implication all translates to an instantiation of the product rule of the Calculus of Inductive of Constructions. Introduction rules can be implemented as abstractions while elimination rules as applications. Finally, type operators can be encoded as parameters of type : \(Type_1 \to \cdots \to Type_1\).  As an example, we show the result of our reflexivity proof from \sttabd{} to Coq\footnote{\(Type_0\) is also denoted Prop in Coq}. Using Coq floating universes, we omit indices for universes. The equality \(=_{\mathcal{L}}\) will be translated as

  \begin{minted}{coq}
    Definition =_L : forall X:Type, X -> X -> Prop :=
      fun (X:Type) (x y:X) =>
         forall (P:X -> Prop), P x -> P y.
  \end{minted}
  while the proof of reflexivity will be translated as the following definition

  \begin{minted}{coq}
    Definition refl_= : forall X:Type, forall x:X, x =_L x :=
      fun X:Type => fun x:X => fun h:(P x) => h.
  \end{minted}

\section{The arithmetic library}
\label{sec:library}
We have implemented these transformations to an arithmetic library that comes from Matita~\cite{matitaarith}. From this library, we have extracted all the lemmas needed to prove the Fermat's little theorem (about \(300\) lemmas). In this library, we can find basic definitions of operators such as \(+,\times\) but also the definition of a permutation over natural numbers or the definition of \textit{big} operator such as \(\Sigma\) or \(\Pi\). This library also proves basic results related to these definitions such as the commutativity of \(+\) or basic results related to prime numbers. In table~\ref{tab:results}, we give some results related to the export of this library to OpenTheory, Coq and Matita.

\begin{table}
  \centering
  \begin{tabular}{ccccc}
              & Dedukti[STT] & OpenTheory & Coq & Matita\\
    size (mb) & 1.5 & 41 & 0.6 & 0.6 \\
    translation time (s) & - & 18 & 3 & 3 \\
    checking time (s) & 0.1 & 13 & 6 & 2
  \end{tabular}
  \caption{Arithmetic library translation}
  \label{tab:results}
\end{table}

These results show that the type checking time in OpenTheory is longer than in Dedukti, Coq or Matita. We suppose that this is mostly due to making the \(\beta\) and \(\delta\) conversions explicit. In order to illustrate the usability of the translated library, we give below the translation of Fermat's little theorem in Coq:

\begin{minted}{coq}
Definition congruent_exp_pred_SO :
  forall p a : nat,
  prime p -> Not (divides p a) -> congruent (exp a (pred p)) (S O) p.
\end{minted}

The constants \mintinline{coq}{prime}, \mintinline{coq}{congruent} and \mintinline{coq}{pred} come with a definition while the constants \mintinline{coq}{exp}, \mintinline{coq}{Not}, \mintinline{coq}{O} and \mintinline{coq}{S} are axiomatized and should be defined by the user. Our tool produces a \textit{functor} that the user should instantiate whose parameters are the axiomatization of those notions. The user should instantiate it with reasonable definitions, proving the axioms. Then the theorem is ready to use. For example, the definition of \mintinline{coq}{exp} has to satisfy the two following axioms:
\begin{minted}{coq}
Axiom sym_eq_exp_body_0 : forall n : nat, (S O) = (exp n O).
Axiom sym_eq_exp_body_S : forall n m : nat, (times (exp n m) n) = (exp n (S m)).
\end{minted}

The following definition (that comes from the standard library) satisfy those definitions:

\begin{minted}{coq}
Fixpoint exp (n m : nat) : Datatypes.nat :=
  match m with
  | O => S O
  | S m => n * exp n m
  end
\end{minted}

For this arithmetic library, one has to define about \(40\) constants and prove about \(80\) axioms. All the constants definitions can be guessed from their name or from the axioms they have to satisfy, and hence the axioms are then easy to prove. This instantiation has been made in Coq~\cite{sttforall}.

 \section{Related Work}
Cauderlier and Dubois already used Dedukti for interoperability in~\cite{DBLP:conf/itp/CauderlierD17}. Their goal was to prove the sieve of Eratosthenes using HOL and Coq in combination. The main advantage of their work is that there is no need to export proofs outside the logical framework, instead everything is checked in Dedukti. However, mathematical objects in Dedukti, such as natural numbers, may have different representation, and therefore this approach may require theorems to transfer results about one representation to results about another representation.

In~\cite{DBLP:conf/itp/KellerW10}, Keller and Werner made a translation from HOL Light to Coq. Despite the fact that their source logic and their target logic is different from ours, they did not use any logical framework.

 OpenTheory~\cite{hurd2011} in itself is an interoperability tool between the HOL family provers. However, OpenTheory is focused for systems that all implement a variant of Higher-Order Logic while this work aims to be more general.

 Beluga~\cite{Pientka:POPL08} is an extension of LF that handles open terms thanks to contextual types. Beluga aims to be useful for interoperability since it is easier to write proof transformations in it.

 The Foundational Proof Certificate project~\cite{DBLP:conf/cade/ChihaniMR13} aims at defining a generic methods for checking proofs. The approach seems more tuned towards self-contained proofs produced by, e.g.,~automated theorem provers, rather than libraries developed in proof assistants and rich logics developed in the rich logics of proof assistants.

 \section{Conclusion}
 \label{sec:conclusion}
In this paper, we showed how \sttabd{} is a simple logic that can be easily represented in the logical framework Dedukti and is powerful enough to express arithmetic proofs. We defined translations from \sttabd{} to other systems such as OpenTheory and implemented these translations from Dedukti. We applied it to an arithmetic library containing a proof of Fermat's little theorem. The differences between OpenTheory and \sttabd{} reveal three difficulties which we addressed in different phases of the translation. In contrast, we showed how the translation to Coq and Matita is easy since \sttabd{} can be seen as a subsystem of the Calculus of Inductive Constructions.

We would like to export this library to other proof systems such as PVS or Agda. While for Agda, the translation should be similar to the one of Coq or Matita, for PVS this is a challenge since there is no proof term but only tactics. In other word, each rule should be translated by an application of one or more tactics. We are also interested to import more proofs in Dedukti[\sttabd{}] that could then be exported.

Finally, we hope that this work is the beginning of a process that could lead to a standardization of libraries, starting with the arithmetic one (naming conventions, constants definitions or statement of important lemmas).


\bibliographystyle{eptcs}
\bibliography{bibliography}

\ifundefinedcolor{longpaper}{}
{
\newpage
\section{Annexes}

  \subsection*{Proof of theorem  \ref{thm:sttsound}}
  \begin{proof}
    To prove the consistency of that logic, we construct a model of \sttabd{} in set theory such that if \(t : T\) then \( [t] \in \llbracket T \rrbracket\). The interpretation of \(\Prop\) is a set of two elements \(\{0,1\}\). Then we show that all the rules preserve the truth-value. That is if \(\Gamma \vdash t\) is derivable and and if \(\forall x\in \Gamma, [x] =1\) then \([t]=1\). Under such interpretation, we will see that \([\bot]=0\), hence the consistency of the logic will be proven because if \(\vdash \bot\) can be proved, then \([\bot]=1\) that would lead to a contradiction.

  First, we define \(\mathcal{M}_f\), the interpretations of any monotype, as the smallest set that contains \(\{42\}\), \(\{0,1\}\) and closed under \(\to\). To avoid problems related to names, we suppose that the set of type variables is countable and that all the names are different. We define \(\llbracket \phantom{A} \rrbracket^{\circ}\) inductively on monotypes:

  \begin{align*}
    \llbracket X \rrbracket_{\rho}^{\circ} &\defn& \rho(X)\\
    \llbracket A \to B \rrbracket_{\rho}^{\circ} &\defn& \llbracket A \rrbracket_{\rho}^{\circ} \to \llbracket B \rrbracket_{\rho}^{\circ}\\
    \llbracket \Prop \rrbracket_{\rho}^{\circ} &\defn& \{0,1\}\\
    \llbracket p~A_1~\dots~A_n \rrbracket_{\rho}^{\circ} &\defn& \{42\}\\
  \end{align*}

  Forall set \(S\) and variable \(X\), we denote \(F(S)_{X}\) the set of functions of domain \(\mathcal{M}_f\) such that \( M \mapsto \llbracket S \rrbracket_{\rho, X \to M}\) . Then we define \(\llbracket \phantom{A} \rrbracket\) on polytypes:
  \begin{align*}
    \llbracket A \rrbracket_{\rho} &\defn& \llbracket A \rrbracket_{\rho}^{\circ}\\
    \llbracket \forallT{X}{}{T} \rrbracket_{\rho} &\defn& F(T)_{X}\\
  \end{align*}

  For the terms, we define \([\phantom{A}]^{\circ}\) and \([\phantom{A}]\). The interpretation uses two functions \(\widetilde{\forall} : \mathcal{P}_{\{0,1\}}^* \mapsto \{0,1\}\)  and \(\widetilde{\Rightarrow} : \{0,1\} \times \{0,1\} \times \{0,1\}\) that are defined as expected.
\begin{align*}
\end{align*}

\begin{align*}
  [x]_{\rho,\sigma}^{\circ} &\defn& \sigma(x)\\
  [cst]_{\rho,\sigma}^{\circ} & \defn&
                               \begin{cases}
                                 [t]_{id,id}^{\circ} \text{ when } cst = t \in \Sigma\\
                                 choice(T) \text{ when } cst : T \in \Sigma\\
                               \end{cases}\\
  [t~u]_{\rho, \sigma}^{\circ} & \defn& [t]_{\rho,\sigma}^{\circ}([u]_{\rho,\sigma}^{\circ})\\
  [\absT{x}{A}{t}]_{\rho,\sigma}^{\circ} &\defn & v \in \llbracket A \rrbracket_{\rho} \mapsto [t]_{\rho,(\sigma, x \mapsto v)}^{\circ}\\
  [t \Rightarrow u]_{\rho,\sigma}^{\circ} &\defn & \widetilde{\Rightarrow} [t]_{\rho,\sigma}^{\circ} [u]_{\rho,\sigma}^{\circ}\\
  [\forallT{x}{A}{t}]_{\rho,\sigma}^{\circ} &\defn& \widetilde{\forall} \{ [t]_{\rho,(\sigma, x \mapsto v)}^{\circ} \mid v \in \llbracket A \rrbracket_{\rho}\}\\
  [\absT{X}{}{t}]_{\rho,\sigma}^{\circ} &\defn& M_X \in \mathcal{M}_f \mapsto [t]_{(\rho,X \mapsto M_X), \sigma}^{\circ}\\
\end{align*}

\begin{align*}
  [ t]_{\rho} & \defn & [t]_{\rho, id}^{\circ}\\
  [ \forallP{A}{}{\tau}]&\defn& \widetilde{\forall} \{[\tau]_{\rho,X \mapsto M_X} \mid M_X \in \mathcal{M}_f\} \\
\end{align*}

Note that the choice function can be defined inductively on the construction of \(\mathcal{M}_F\).

We can now prove the following lemma:

\begin{lemma}
  If \(\mathcal{C} \vdash t : T\) then
  \([t]_{[\mathcal{C}]} \in \llbracket T \rrbracket_{[\mathcal{C}]} \)
\end{lemma}
\begin{proof}
  By induction on the derivation of \(t\).
\end{proof}

One can also check the following properties:
\begin{itemize}
\item the interpretation of \(\forallT{x}{\Prop}{x}\) is \(0\)
\item the interpretation of \(\forallT{x}{\Prop}{x \Rightarrow x}\) is \(1\)
\item if \(t \rightarrow_{\beta} u\) then \([t]=[u]\)
\item if \(t \rightarrow_{\delta} u\) then \([t]=[u]\).
\item if \(t \equiv_{\beta,\delta} u\) then \([t] = [u]\).
\end{itemize}
\begin{lemma}
  If \(\mathcal{C} \vdash t\) and assuming for every \(x:\Prop \in \mathcal{C}\), \([x]=1\) then
  \([t]_{[\mathcal{C}]} = 1 \)
\end{lemma}
\begin{proof}
  By induction on the derivation of \(t\).
\end{proof}

From the previous lemma, we can conclude that it is not possible to proof \(\bot\) in the empty context since its interpretation is \(0\). Hence, \sttabd{} is consistent.
\end{proof}
}

\ifundefinedcolor{longpaper}{}{
  \subsection*{Proof of theorem~\ref{thm:sttdecidable}}
  \begin{proof}
To prove that the type checking and the proof checking are decidable in \sttabd{}, we only have to prove that \(\equiv_{\beta\delta}\) is decidable. This property follows from the fact that \(\hookrightarrow_{\beta} \cup \hookrightarrow_{\delta}\) is a convergent term rewriting system (see Theorem.~\ref{thm:convergent}). It follows that testing that two terms are congruent is the same as computing their normal form.
  \end{proof}
\begin{theorem}
  \label{thm:convergent}
  \(\hookrightarrow_{\beta} \cup \hookrightarrow_{\delta}\) is a convergent term rewriting system.
\end{theorem}
\begin{proof}
  One can see that the type system of \sttabd{} is almost a sub-system of system \(F_{\omega}\). Indeed, \(\Prop\) can be interpreted as a type constructor. Types operators can be encoded in the system \(F_{\omega}\) via the following encoding:
  \begin{align*}
    \llbracket p \rrbracket &= A_1 \to \dots \to A_n \to p
  \end{align*}
where \(p\) is a type constructor of arity \(n\). Hence, the convergence of \(\hookrightarrow_{\beta\delta}\) can be proven using the one from the calculus of system \(F_{\omega}\) with constants.
\end{proof}
}

\ifundefinedcolor{longpaper}{}{
  \subsection*{Translation functions from \sttabd{} to Dedukti[\sttabd{}]}
\begin{definition}
  We define the function \(|\cdot|_{mty}\) that takes a \sttabd{} monotype to a term of \lpcstt{} inductively as:
  \begin{align*}
    |X|_{mty} &=& X\\
    |\Prop|_{mty} &=& prop\\
    |A \to B|_{mty} &=& |A|_{mty}\dot{\to}|B|_{mty}\\
    |p(A_1,\dots,A_n)|_{mty} &=& p~|A_1|_{mty}~\dots~|A_n|_{mty}
  \end{align*}

  then we define the function \(|\cdot|_{pty}\) that takes a \sttabd{} type to a term of \lpcstt{} inductively as:
  \begin{align*}
    |A|_{pty} &=& p~|A|_{mty}\\
    |\forall X.~T|_{pty} &=& forallK~(\absT{X}{pty}{|T|_{pty})}
  \end{align*}
\end{definition}

\begin{definition}
  We define the function \(|\cdot|_{cctx}\) that takes a \sttabd{} constant term to a term of \lpcstt{} inductively as:
  \begin{align*}
    |cst|_{cst} &=& cst\\
    |c~A|_{cst} &=& |c|_{cst}~|A|_{pty}
  \end{align*}
\end{definition}
\begin{definition}
  We define the function \(|\cdot|_{mte}\) that takes a \sttabd{} monoterm to a term of \lpcstt{} inductively as:
  \begin{align*}
    |x|_{mte} &=& x\\
    |c|_{mte} &=& |c|_{cst}\\
    |\app{t_1}{t_2}|_{mte} &=& |t_1|_{mte}~|t_2|_{mte}\\
    |\absT{x}{A}{t}|_{mte} &=& \absT{x}{eta~|A|_{pty}}{|t|_{mte}}\\
    | t_1 \Rightarrow t_2|_{mte} &=& impl~|t_1|_{mte}~|t_2|_{mte}\\
    | \forall x^A, t|_{mte} &=& forall~|A|_{pty}~(\absT{x}{eta~|A|_{pty}}{|t|_{mte}})\\
    |\abs{A}{t}|_{mte} &=& \absT{A}{pty}{|t|_{mte}}
  \end{align*}
\end{definition}
\begin{definition}
  We define the function \(|\cdot|_{pte}\) that takes a \sttabd{} term to a term of \lpcstt{} inductively as:
  \begin{align*}
    |t|_{pte} &=& |t|_{mte}\\
    |\forall A,\tau|_{pte} &=& forallP~(\absT{A}{pty}{|\tau|_{pte})}
  \end{align*}
\end{definition}

\begin{definition}
    We define the function \(|\cdot|_{cctx}\) that takes a \sttabd{} constant context to a context of \lpcstt{} inductively as:
  \begin{align*}
    |\emptyset|_{cctx} &=& \Sigma_{\text{\sttabd}}\\
    |\Sigma, cst:T|_{cctx} &=& |\Sigma|_{cctx} \cup (cst :eta~|T|_{pty})\\
    |\Sigma, cst = \tau:T|_{cctx} &=& |\Sigma|_{cctx} \cup (cst :eta~|T|_{pty}) \cup (cst \hookrightarrow |\tau|_{pte})
  \end{align*}
\end{definition}

\begin{definition}
  We define the function \(|\cdot|_{tectx}\) that takes a \sttabd{} tern context to a context of \lpcstt{} inductively as:
  \begin{align*}
    |\emptyset|_{tectx} &=& \emptyset\\
    |\Gamma, x:A|_{tectx} &=& |\Gamma|_{tectx} \cup x:eta~|A|_{pty}\\
    |\Gamma, X|_{tectx} &=& |\Gamma|_{tectx} \cup X:type
  \end{align*}
\end{definition}
\begin{definition}
  We define the function \(\left[\phantom{P}\right]_{proof}^{ctx}\) that translates a derivation tree to a Dedukti[\sttabd{}] term. The function is indexed by a \(ctx\) to remember the hypothesis introduced during the proof for the \(assume\) case. In the following, the notation \(\mathcal{P}_t\) means that \(t\) is the proposition (term) proved by the derivation \(\mathcal{P}\).
  \begin{align*}
    \left[\Rule{ t \vdash t}{S assume}\right]_{proof}^{ctx} &=& find~t~ctx\\
    \left[\RulePP{ \mathcal{P}_{t} }{\mathcal{P}_{t\rightarrow u}}{\mathcal{C} \vdash u }{S \(\Rightarrow_E\)}\right]_{proof}^{ctx} &=& \left[\mathcal{P}_{t\rightarrow u}\right]_{proof}^{ctx}~\left[\mathcal{P}_{t}\right]_{proof}^{ctx}\\
    \left[\RuleP{ \mathcal{P}_{u}}{\mathcal{C} \vdash t \Rightarrow u }{S \(\Rightarrow_I\)}\right]_{proof}^{ctx} &=& \absT{x}{proof~|t|_{pte}}{\left[\mathcal{P}_{u}\right]_{proof}^{(t,x)::ctx}} \\
    \left[\RulePP{\mathcal{P}_{\forall t} }{\mathcal{J}_{u}}{\mathcal{C} \vdash t[u/x] }{S \(\forall_E\)}\right]_{proof}^{ctx} &=& \left[\mathcal{P}_{\forall t}\right]_{proof}^{ctx}~\left|u\right|_{pte}\\
    \left[\RuleP{ \mathcal{P}_t}{\mathcal{C} \vdash \forall x^A, t }{S \(\forall_I\)}\right]_{proof}^{ctx} &=& \absT{x}{eta~|t|_{pte}}{\left[\mathcal{P}_{t}\right]_{proof}^{ctx}} \\ \\
    \left[\RuleP{ \mathcal{P}_t }{\mathcal{C} \vdash u }{S conv}\right]_{proof}^{ctx} &=& \mathcal{P}_t\\
     \left[\RuleP{ \mathcal{P}_{\tau}}{\mathcal{C} \vdash \forall X,\tau}{S \(\forallp_E\)}\right]_{proof}^{ctx} &=& \absT{X}{pty}{\left[\mathcal{P}_{\tau}\right]_{proof}^{ctx}} \\
  \end{align*}

  For the translation of \(\Rightarrow_{I}\), \(x\) should be a fresh variable.
\end{definition}

}
\ifundefinedcolor{longpaper}
{
}
{
  \subsection*{Proof that the encoding is sound}
  In the following \(\vdash_D\) means that this is a \lpc{} judgment. We also omit the \(\Sigma_{\text{\sttabd{}}}\) context inside the judgment for clarity as well as the context for constants.

  \begin{lemma}{Free variables in monotypes}
    \label{lemma:fvmonotypes}

    \(\forall \Gamma\), \(A\), if \(\Gamma \vdash A~\mathbf{wf}\) is derivable then \(FV(A) \subseteq \Gamma\)
  \end{lemma}
  \begin{proof}
    By induction on \(A\), all cases are trivial.
  \end{proof}

  \begin{lemma}{Translation of monotypes}
    \label{lemma:monotypes}

    Forall \(\Gamma\) and \(A\), if \(\Gamma \vdash A~\mathbf{wf}\) is derivable, then
    \[|\Gamma | \vdash |A|_{mty} : type\]
  \end{lemma}
  \begin{proof}
    By induction on \(A\) using lemma~\ref{lemma:fvmonotypes}.
  \end{proof}

  \begin{lemma}{Translation of polytypes}
    \label{lemma:types}

    \(\forall \Gamma\) and \(T\), if \(\Gamma \vdash T~\mathbf{wf}\) is derivable, then
    \[|\Gamma | \vdash |T|_{pty} : ptype\]
  \end{lemma}

  \begin{proof}
    By induction on \(T\) using lemma~\ref{lemma:monotypes}.
  \end{proof}

  \begin{lemma}{Constant contexts}
    \label{lemma:constants}

    For all constant context \(\Sigma\), \[\Sigma~\mathbf{wf} \Rightarrow |\Sigma|_{cctx}~\mathbf{wf} \]

    The context of the right-hand side being a \lpc{} context.
  \end{lemma}

  \begin{lemma}[Context]
    \label{lemma:sdcontexts}
    Forall \(\Gamma\), if \(\vdash \Gamma~\mathbf{wf}\) is derivable, then \(\vdash |\Gamma|~\mathbf{wf}\) is derivable.
  \end{lemma}

  \begin{lemma}{Translation of terms schema}
    \label{lemma:termschema}

    For every term schema \(t\) and type schema \(A\),
    \[\mathcal{C} \vdash t : A \Rightarrow \Gamma_t \vdash |t|_{mte} : eta~|A|_{pty}\]
    where \(\Gamma_t\) contains declarations of free variables and free type variables that occur in \(t\).
  \end{lemma}

  \begin{proof}
    Lemma~\ref{lemma:constants}, lemma~\ref{lemma:sdcontexts} and lemma~\ref{lemma:termschema} are proved by mutual recursion on the derivation of each of judgment.
  \end{proof}

  \begin{theorem}[Soundness]
    If \(\mathcal{C} \vdash t\) then \(|\mathcal{C}| \vdash |\mathcal{C} \vdash t| : proof~|t|\)
  \end{theorem}

  \begin{proof}
    By induction on the derivation of \(|\mathcal{C} \vdash t|\).
  \end{proof}
}

\ifundefinedcolor{longpaper}{}
{
  \subsection*{Encoding \sttabd{} connectives from equality}
To show the admissibility of the other rules, we need to prove first the admissibility of the three rules below:
\begin{mathpar}
  \addrule{\RuleP{\Gamma \vdash \app{(\abs{x}{t_1})}{u_1} =  \app{(\abs{x}{t_2})}{u_2}}{\Gamma \vdash \subst{t_1}{x}{u_1} = \subst{t_2}{x}{u_2}}{\(\beta_{=}\)}}
  \addrule{\RuleP{\Gamma \vdash x \land y}{\Gamma \vdash x}{\(\land_1\)}}
  \addrule{\RuleP{\Gamma \vdash x \land y}{\Gamma \vdash y}{\(\land_2\)}}
\end{mathpar}

\paragraph{\(\beta_{=}\)}

\begin{mathpar}
  \scriptsize
  \inferrule* [right=trans,sep=8pt]
  {
    \inferrule*
    { }
    { \vdash (\lambda y.~t_2)~u_2 = \subst{t_2}{x}{u_2}}
    \\
    \inferrule*[right=trans]
    {
      \inferrule*
      {}
      {\Gamma \vdash (\lambda x.~t_1)~u_1 = (\lambda x.~t_2)~u_2}
      \\
      \inferrule*
      {
        \inferrule*
        { }
        {\vdash \app{(\abs{x}{t_1})}{u_1} = \subst{t_1}{x}{u_1}}
      }
      {\vdash \subst{t_1}{x}{u_1} = \app{(\abs{x}{t_1})}{u_1}}
    }
    {\Gamma \subst{t_1}{x}{u_1} = (\lambda x.~t_2)~u_2}
  }
  {\Gamma \vdash \subst{t_1}{x}{u_1} = \subst{t_2}{x}{u_2}}
\end{mathpar}

\paragraph{\(\land_1\)}
From the definition of \(\land\), one can use the projection \(\lambda x.~\lambda y.~x\) to get the first component.

\paragraph{\(\land_2\)}
Same proof except we use the function \(\lambda x.~\lambda y.~y\)

\paragraph{\(\Rightarrow_E\)}

\begin{mathpar}
  \scriptsize
  \inferrule*
  {
    \inferrule*
    {
      \inferrule* [Right=\(\land_2\)]
      {
        \inferrule* [Right=eqMp]
        {
          \inferrule*
          { }
          {p \vdash p}
          \\
          \inferrule* [Right=sym]
          {
            \inferrule* [Right=eqMp]
            {
              \inferrule*
              { }
              {p \Rightarrow q \vdash p \Rightarrow q}
              \\
              \inferrule* [Right=instantation]
              {
                \inferrule* [Right=Axiom]
                { }
                {\vdash \lambda x. \lambda y.~x \Rightarrow y = x \land y = x}
              }
              {\vdash p \Rightarrow q = p \land q = p}
            }
            {p \Rightarrow q \vdash p \land q = p}
          }
          {p = p \land q}
        }
        {p, p\Rightarrow q \vdash p \land q}
      }
      {p, p\Rightarrow q \vdash q}
      \\
      \inferrule*
      { \Pi}
      {\Gamma \vdash p}
    }
    { \Gamma, p \Rightarrow q \vdash q}
    \\
    \inferrule*
    { \Pi}
    {\Gamma \vdash p \Rightarrow q}
  }
  { \Gamma \vdash q}

\end{mathpar}

where \textit{instantiation} is the following proof

\begin{mathpar}
  \inferrule* [Right=\(\beta_{=}\)]
  {
    \inferrule* [Right=App]
    {
      \inferrule*  [Right=\(\beta_{=}\)]
      {
        \inferrule* [Right=App]
        {
          \inferrule* [right=Axiom]
          { }
          {\vdash \lambda x. \lambda y.~x \Rightarrow y = \lambda x. \lambda y.~x \land y = x}
          \\
          \inferrule* [Right=Refl]
          { }
          {\vdash p = p}
        }
        {\vdash (\lambda x.\lambda y.~x \Rightarrow y)~p = (\lambda x.\lambda y.~x \land y = x)~p}
      }
      {\vdash \lambda y.~p \Rightarrow y = \lambda y.~p \land y = p}
      \\
      \inferrule* [Right=Refl]
      { }
      {\vdash q = q}
    }
    {\vdash (\lambda y.~p \Rightarrow y)~q = (\lambda y.~p \land y = p)~q}
  }
  {\vdash p \Rightarrow q = p \land q = p}
\end{mathpar}

\paragraph{\(\forall_I\)}

\begin{mathpar}
  \inferrule*[right=EqMp]
  {
    \inferrule*
    {
      \inferrule*[right=DeductAntiSym]
      {
        \inferrule*
        {
          \Pi
        }
        { \Gamma \vdash t}
        \\
        \inferrule* { }
        { \Gamma \vdash \top }
      }
      { \Gamma \vdash t = \top }
    }
    {\Gamma \vdash \lambda x.~t = \lambda x.~\top }
    \\
    \inferrule*
    {
      \inferrule*[Right=instant+axiom]
      { }
      {\Gamma \vdash \forall x^A.~t = (\lambda x.~t = \lambda x.~\top)}
    }
    {\Gamma \vdash (\lambda x.~t = \lambda x.~\top) = \forall x^A.~t}
  }
  { \Gamma \vdash \forall x^A.~t}
\end{mathpar}

\paragraph{\(\forall_E\)}

\begin{mathpar}
  \scriptsize
  \inferrule*
  {
    \inferrule*
    {}
    { \Gamma \vdash \top }
    \\
    \inferrule*
    {
      \inferrule*
      {
        \inferrule*
        {
          \inferrule*
          {
            \inferrule*
            {\Pi }
            { \Gamma \vdash \forall x.~t}
            \\
            \inferrule*
            { }
            {\Gamma \vdash \forall x.~t = \left(\lambda x.~t = \lambda x.~\top\right)}
          }
          { \Gamma \vdash \lambda x.~t = \lambda x.~\top }
          \\
          \inferrule*
          { }
          { \vdash u = u }
        }
        { \Gamma \vdash (\lambda x.~t)~u = (\lambda x.~\top)~u}
        \\
        \inferrule*
        {}
        { \Gamma \vdash (\lambda x.~\top)~u = \top}
      }
      {\Gamma \vdash (\lambda x.~t)~u = \top }
    }
    { \Gamma \vdash \top = (\lambda x.~t)~u }
  }
  { \Gamma \vdash (\lambda x.~t)~u }
\end{mathpar}
\begin{mathpar}
  \inferrule*
  {
    \inferrule*
    { \vdots }
    {\Gamma \vdash (\lambda x.~t)~u }
    \\
    \inferrule*[Right=\(\beta\)]
    { }
    { \Gamma \vdash (\lambda x^A.~t)~u = t[x := u]}
  }
  { \Gamma \vdash t[x := u] }
\end{mathpar}

}

  \ifundefinedcolor{longpaper}{}
  {
    \subsection*{Contextual rules in OpenTheory}
      \begin{mathpar}
        \inferrule*
        {
          \inferrule*
          {}
          {\Gamma \vdash t = t'}
        }
        {\Gamma \vdash \forallP{X}{}{t} = \forallP{X}{}{t'}}
      \end{mathpar}

      \begin{mathpar}
        \inferrule*[Left=OT DeductAntiSym]
        {
          \inferrule*[Left=S \(\forallp\) intro]
          {
            \inferrule*[Left=OT EqMp]
            {
              \inferrule*[Left=S \(\forallp\) elim]
              {
                \inferrule*[Left=OT Assume]
                { }
                {X; \Gamma, \forallP{X}{}{t} \vdash \forallP{X}{}{t} }
              }
              {X; \Gamma, \forallP{X}{}{t} \vdash t}
              \\
              \inferrule*
              {}
              {\Gamma \vdash t = t'}
            }
            {X; \Gamma, \forallP{X}{}{t} \vdash t'}
          }
          {\Gamma, \forallP{X}{}{t} \vdash \forallP{X}{}{t'}}
          \\
          \inferrule*
          {}
          {\vdots}
        }
        {\Gamma \vdash \forallP{X}{}{t} = \forallP{X}{}{t'}}
      \end{mathpar}

    The rules \textsc{S \(\forallp_E\)} and \textsc{S \(\forallp_I\)} will be eliminated at section~\ref{sec:quantifiers}.

    \begin{mathpar}
        \inferrule*
        {
          \inferrule*
          {}
          {\Gamma \vdash t = t'}
        }
        {\Gamma \vdash \forallT{x}{A}{t} = \forallT{x}{A}{t'}}
      \end{mathpar}

            \begin{mathpar}
        \inferrule*[Left=OT DeductAntiSym]
        {
          \inferrule*[Left=S \(\forall\) intro]
          {
            \inferrule*[Left=OT EqMp]
            {
              \inferrule*[Left=S \(\forall\) elim]
              {
                \inferrule*[Left=OT Assume]
                { }
                {x:A; \Gamma, \forallT{x}{A}{t} \vdash \forallT{x}{A}{t} }
              }
              {x:A; \Gamma, \forallT{x}{A}{t} \vdash t}
              \\
              \inferrule*
              {}
              {\Gamma \vdash t = t'}
            }
            {x:A; \Gamma, \forallT{x}{A}{t} \vdash t'}
          }
          {\Gamma, \forallT{x}{A}{t} \vdash \forallT{x}{A}{t'}}
          \\
          \inferrule*
          {}
          {\vdots}
        }
        {\Gamma \vdash \forall{x}{A}{t} = \forall{x}{A}{t'}}
      \end{mathpar}

      The rules \textsc{S \(\forall_E\)} and \textsc{S \(\forall_I\)} are encoded as shown in section~\ref{sec:forallimpl}.
    }
\end{document}

%% file: prelude.tex
\input{packages.tex}
\input{tikz.tex}

\newtheorem{theorem}{Theorem}

\newtheorem{lemma}{Lemma}

\newtheorem{definition}{Definition}

\newcommand{\ifundefinedcolor}[3]{\ifundefined{#1}{\color{black}{#2}\color{black}}{\color{black}{#3}\color{black}}}

\newcommand{\dkcolor}{Bittersweet}

\newcommand{\otcolor}{RedViolet}

\newcommand{\sttcolor}{OliveGreen}

\newcommand{\dkstyle}[1]{\textcolor{\dkcolor}{#1}}

\newcommand{\otstyle}[1]{\textcolor{\otcolor}{#1}}

\newcommand{\sttstyle}[1]{\textcolor{\sttcolor}{#1}}

\newcommand{\sttabd}{STT\(\forall_{\beta\delta}\)}

\newcommand{\qzero}{\ensuremath{\mathcal{Q}_0}}

\newcommand{\rw}{\hookrightarrow}

\newcommand{\forallT}[3]{\ensuremath{\forall #1^{#2}.~#3}}

\newcommand{\forallp}[0]{\rotatebox[origin=c]{180}{\ensuremath{\mathcal{A}}}}

\newcommand{\forallP}[2]{\ensuremath{\forallp #1.~#2}}

\newcommand{\defn}{:\equiv}

\newcommand{\aff}{:=}

\newcommand{\app}[2]{#1~#2}

\newcommand{\absT}[3]{\lambda#1^{#2}.~#3}

\newcommand{\abs}[2]{\absT{#1}{}{#2}}

\newcommand{\product}[3]{\Pi#1:#2.~#3}

\newcommand{\subst}[3]{#1[#2\aff#3]}

\newcommand{\Prop}{\mathbf{Prop}}

\newcommand{\ot}{\systemnamestyle{OpenTheory}}

\newcommand{\lp}{\(\lambda\Pi\)}

\newcommand{\lpc}{\ensuremath{\lambda\Pi\text{-calculus modulo theory}}}

\newcommand{\lpcstt}{Dedukti[\sttabd{}]}

\newcommand{\systemnamestyle}{\textsc}

\newcommand{\Rule}[2]{\inferrule*[right=#2]{ }{#1}}

\newcommand{\RuleP}[3]{\inferrule*[right=#3]{#1}{#2}}

\newcommand{\RulePP}[4]{\inferrule*[right=#4]{#1 \\ #2}{#3}}

\newcommand{\RulePPP}[5]{\inferrule*[right=#5]{#1 \\ #2 \\ #3}{#4}}

\newcommand{\addrule}[1]{#1 \and}

\newcommand{\ifundefined}[3]{\expandafter\ifx\csname #1\endcsname\relax #2\else #3\fi}
\newcommand{\benkhide}[0]{\def \@benkhide{}}

\newenvironment{rules}[2]
{
  \def\tmpvariable{#1}
  \def\tmpvariabled{#2}
  \begin{figure}
    \begin{mdframed}
      \begin{center}
  \begin{mathpar}
}
{
\end{mathpar}
\end{center}

\caption{\tmpvariable}
\label{\tmpvariabled}
\end{mdframed}
\end{figure}

}

\newcommand{\nil}[0]{\ensuremath{\mathrm{[~]}}}
\newcommand{\Kind}[0]{\ensuremath{\mathbf{Kind}}}
\newcommand{\Type}[0]{\ensuremath{\mathbf{Type}}}
\newcommand{\wf}[0]{\ensuremath{\text{well-formed}}}
\newcommand{\PiT}[3]{\ensuremath{\Pi{}#1^{#2}.~#3}}


%% file: packages.tex
\usepackage{microtype}

\usepackage{rotating}

\usepackage[T1]{fontenc}

\usepackage[utf8]{inputenc}

\usepackage{lmodern}

\usepackage{graphicx}

\usepackage{latexsym}

\usepackage{hyperref}

\usepackage{mdframed}

\usepackage{caption}

\usepackage{tikz}

\usepackage{amsthm}

\usepackage{amsmath}

\usepackage{amssymb}

\usepackage{amsfonts}

\usepackage{bussproofs}

\usepackage{mathpartir}

\usepackage{subcaption}

\usepackage{stmaryrd}

\usepackage[frozencache]{minted}

\usepackage{xcolor}

\usepackage{wrapfig}

\usepackage{bussproofs}

%% file: tikz.tex
\usetikzlibrary{arrows,shapes,decorations.pathmorphing,backgrounds,positioning,fit,calc}

\definecolor{blue-violet}{rgb}{0.54, 0.17, 0.89}

%% file: sttforall-lfmtp.bbl
\begin{thebibliography}{10}
\providecommand{\bibitemdeclare}[2]{}
\providecommand{\surnamestart}{}
\providecommand{\surnameend}{}
\providecommand{\urlprefix}{Available at }
\providecommand{\url}[1]{\texttt{#1}}
\providecommand{\href}[2]{\texttt{#2}}
\providecommand{\urlalt}[2]{\href{#1}{#2}}
\providecommand{\doi}[1]{doi:\urlalt{http://dx.doi.org/#1}{#1}}
\providecommand{\bibinfo}[2]{#2}

\bibitemdeclare{book}{DBLP:books/daglib/0070479}
\bibitem{DBLP:books/daglib/0070479}
\bibinfo{author}{Peter~B. \surnamestart Andrews\surnameend}
  (\bibinfo{year}{1986}): \emph{\bibinfo{title}{An introduction to mathematical
  logic and type theory - to truth through proof}}.
\newblock \bibinfo{series}{Computer science and applied mathematics},
  \bibinfo{publisher}{Academic Press}.

\bibitemdeclare{inproceedings}{DBLP:conf/cade/AspertiRCT11}
\bibitem{DBLP:conf/cade/AspertiRCT11}
\bibinfo{author}{Andrea \surnamestart Asperti\surnameend},
  \bibinfo{author}{Wilmer \surnamestart Ricciotti\surnameend},
  \bibinfo{author}{Claudio~Sacerdoti \surnamestart Coen\surnameend} \&
  \bibinfo{author}{Enrico \surnamestart Tassi\surnameend}
  (\bibinfo{year}{2011}): \emph{\bibinfo{title}{The Matita Interactive Theorem
  Prover}}.
\newblock In \bibinfo{editor}{Nikolaj \surnamestart Bj{\o}rner\surnameend} \&
  \bibinfo{editor}{Viorica \surnamestart Sofronie{-}Stokkermans\surnameend},
  editors: {\sl \bibinfo{booktitle}{Automated Deduction - {CADE-23} - 23rd
  International Conference on Automated Deduction, Wroclaw, Poland, July 31 -
  August 5, 2011. Proceedings}}, {\sl \bibinfo{series}{Lecture Notes in
  Computer Science}} \bibinfo{volume}{6803}, \bibinfo{publisher}{Springer}, pp.
  \bibinfo{pages}{64--69}.
\newblock \urlprefix\url{https://doi.org/10.1007/978-3-642-22438-6_7}.

\bibitemdeclare{inproceedings}{DBLP:conf/itp/CauderlierD17}
\bibitem{DBLP:conf/itp/CauderlierD17}
\bibinfo{author}{Rapha{\"{e}}l \surnamestart Cauderlier\surnameend} \&
  \bibinfo{author}{Catherine \surnamestart Dubois\surnameend}
  (\bibinfo{year}{2017}): \emph{\bibinfo{title}{FoCaLiZe and Dedukti to the
  Rescue for Proof Interoperability}}.
\newblock In \bibinfo{editor}{Mauricio \surnamestart
  Ayala{-}Rinc{\'{o}}n\surnameend} \& \bibinfo{editor}{C{\'{e}}sar~A.
  \surnamestart Mu{\~{n}}oz\surnameend}, editors: {\sl
  \bibinfo{booktitle}{Interactive Theorem Proving - 8th International
  Conference, {ITP} 2017, Bras{\'{\i}}lia, Brazil, September 26-29, 2017,
  Proceedings}}, {\sl \bibinfo{series}{Lecture Notes in Computer Science}}
  \bibinfo{volume}{10499}, \bibinfo{publisher}{Springer}, pp.
  \bibinfo{pages}{131--147}.
\newblock \urlprefix\url{https://doi.org/10.1007/978-3-319-66107-0_9}.

\bibitemdeclare{inproceedings}{DBLP:conf/cade/ChihaniMR13}
\bibitem{DBLP:conf/cade/ChihaniMR13}
\bibinfo{author}{Zakaria \surnamestart Chihani\surnameend},
  \bibinfo{author}{Dale \surnamestart Miller\surnameend} \&
  \bibinfo{author}{Fabien \surnamestart Renaud\surnameend}
  (\bibinfo{year}{2013}): \emph{\bibinfo{title}{Foundational Proof Certificates
  in First-Order Logic}}.
\newblock In \bibinfo{editor}{Maria~Paola \surnamestart Bonacina\surnameend},
  editor: {\sl \bibinfo{booktitle}{Automated Deduction - {CADE-24} - 24th
  International Conference on Automated Deduction, Lake Placid, NY, USA, June
  9-14, 2013. Proceedings}}, {\sl \bibinfo{series}{Lecture Notes in Computer
  Science}} \bibinfo{volume}{7898}, \bibinfo{publisher}{Springer}, pp.
  \bibinfo{pages}{162--177}.
\newblock \urlprefix\url{https://doi.org/10.1007/978-3-642-38574-2_11}.

\bibitemdeclare{inproceedings}{DBLP:conf/lics/Coquand86}
\bibitem{DBLP:conf/lics/Coquand86}
\bibinfo{author}{Thierry \surnamestart Coquand\surnameend}
  (\bibinfo{year}{1986}): \emph{\bibinfo{title}{An Analysis of Girard's
  Paradox}}.
\newblock In: {\sl \bibinfo{booktitle}{Proceedings of the Symposium on Logic in
  Computer Science {(LICS} '86), Cambridge, Massachusetts, USA, June 16-18,
  1986}}, \bibinfo{publisher}{{IEEE} Computer Society}, pp.
  \bibinfo{pages}{227--236}.

\bibitemdeclare{inproceedings}{DBLP:conf/tlca/CousineauD07}
\bibitem{DBLP:conf/tlca/CousineauD07}
\bibinfo{author}{Denis \surnamestart Cousineau\surnameend} \&
  \bibinfo{author}{Gilles \surnamestart Dowek\surnameend}
  (\bibinfo{year}{2007}): \emph{\bibinfo{title}{Embedding Pure Type Systems in
  the Lambda-Pi-Calculus Modulo}}.
\newblock In \bibinfo{editor}{Simona Ronchi~Della \surnamestart
  Rocca\surnameend}, editor: {\sl \bibinfo{booktitle}{Typed Lambda Calculi and
  Applications, 8th International Conference, {TLCA} 2007, Paris, France, June
  26-28, 2007, Proceedings}}, {\sl \bibinfo{series}{Lecture Notes in Computer
  Science}} \bibinfo{volume}{4583}, \bibinfo{publisher}{Springer}, pp.
  \bibinfo{pages}{102--117}.
\newblock \urlprefix\url{http://dx.doi.org/10.1007/978-3-540-73228-0_9}.

\bibitemdeclare{article}{DBLP:journals/jacm/HarperHP93}
\bibitem{DBLP:journals/jacm/HarperHP93}
\bibinfo{author}{Robert \surnamestart Harper\surnameend},
  \bibinfo{author}{Furio \surnamestart Honsell\surnameend} \&
  \bibinfo{author}{Gordon~D. \surnamestart Plotkin\surnameend}
  (\bibinfo{year}{1993}): \emph{\bibinfo{title}{A Framework for Defining
  Logics}}.
\newblock {\sl \bibinfo{journal}{J. {ACM}}}
  \bibinfo{volume}{40}(\bibinfo{number}{1}), pp. \bibinfo{pages}{143--184}.
\newblock \urlprefix\url{http://doi.acm.org/10.1145/138027.138060}.

\bibitemdeclare{inproceedings}{hurd2011}
\bibitem{hurd2011}
\bibinfo{author}{Joe \surnamestart Hurd\surnameend} (\bibinfo{year}{2011}):
  \emph{\bibinfo{title}{The {OpenTheory} Standard Theory Library}}.
\newblock In \bibinfo{editor}{Mihaela \surnamestart Bobaru\surnameend},
  \bibinfo{editor}{Klaus \surnamestart Havelund\surnameend},
  \bibinfo{editor}{Gerard~J. \surnamestart Holzmann\surnameend} \&
  \bibinfo{editor}{Rajeev \surnamestart Joshi\surnameend}, editors: {\sl
  \bibinfo{booktitle}{Third International Symposium on {NASA} Formal Methods
  ({NFM 2011})}}, {\sl \bibinfo{series}{Lecture Notes in Computer Science}}
  \bibinfo{volume}{6617}, \bibinfo{publisher}{Springer}, pp.
  \bibinfo{pages}{177--191}.
\newblock \urlprefix\url{https://doi.org/10.1007/3-540-60275-5_76}.

\bibitemdeclare{inproceedings}{DBLP:conf/itp/KellerW10}
\bibitem{DBLP:conf/itp/KellerW10}
\bibinfo{author}{Chantal \surnamestart Keller\surnameend} \&
  \bibinfo{author}{Benjamin \surnamestart Werner\surnameend}
  (\bibinfo{year}{2010}): \emph{\bibinfo{title}{Importing {HOL} Light into
  Coq}}.
\newblock In \bibinfo{editor}{Matt \surnamestart Kaufmann\surnameend} \&
  \bibinfo{editor}{Lawrence~C. \surnamestart Paulson\surnameend}, editors: {\sl
  \bibinfo{booktitle}{Interactive Theorem Proving, First International
  Conference, {ITP} 2010, Edinburgh, UK, July 11-14, 2010. Proceedings}}, {\sl
  \bibinfo{series}{Lecture Notes in Computer Science}} \bibinfo{volume}{6172},
  \bibinfo{publisher}{Springer}, pp. \bibinfo{pages}{307--322}.
\newblock \urlprefix\url{https://doi.org/10.1007/978-3-642-14052-5_22}.

\bibitemdeclare{article}{Pfenning:1988:HAS:960116.54010}
\bibitem{Pfenning:1988:HAS:960116.54010}
\bibinfo{author}{F.~\surnamestart Pfenning\surnameend} \&
  \bibinfo{author}{C.~\surnamestart Elliott\surnameend} (\bibinfo{year}{1988}):
  \emph{\bibinfo{title}{Higher-order Abstract Syntax}}.
\newblock {\sl \bibinfo{journal}{SIGPLAN Not.}}
  \bibinfo{volume}{23}(\bibinfo{number}{7}), pp. \bibinfo{pages}{199--208}.
\newblock \urlprefix\url{http://doi.acm.org/10.1145/960116.54010}.

\bibitemdeclare{inproceedings}{Pientka:POPL08}
\bibitem{Pientka:POPL08}
\bibinfo{author}{Brigitte \surnamestart Pientka\surnameend}
  (\bibinfo{year}{2008}): \emph{\bibinfo{title}{A type-theoretic foundation for
  programming with higher-order abstract syntax and first-class
  substitutions}}.
\newblock In: {\sl \bibinfo{booktitle}{35th Annual {ACM} Symposium on
  Principles of Programming Languages (POPL'08)}}, \bibinfo{publisher}{ACM},
  pp. \bibinfo{pages}{371--382}.
\newblock \urlprefix\url{https://doi.org/10.1145/1328438.1328483}.

\bibitemdeclare{phdthesis}{DBLP:phd/hal/Saillard15a}
\bibitem{DBLP:phd/hal/Saillard15a}
\bibinfo{author}{Ronan \surnamestart Saillard\surnameend}
  (\bibinfo{year}{2015}): \emph{\bibinfo{title}{Typechecking in the
  lambda-Pi-Calculus Modulo : Theory and Practice. (V{\'{e}}rification de
  typage pour le lambda-Pi-Calcul Modulo : th{\'{e}}orie et pratique)}}.
\newblock Ph.D. thesis, \bibinfo{school}{Mines ParisTech, France}.
\newblock \urlprefix\url{https://tel.archives-ouvertes.fr/tel-01299180}.

\bibitemdeclare{misc}{thecoqdevelopmentteam20171133970}
\bibitem{thecoqdevelopmentteam20171133970}
\bibinfo{author}{The Coq~Development \surnamestart Team\surnameend}
  (\bibinfo{year}{2017}): \emph{\bibinfo{title}{The Coq Proof Assistant,
  version 8.7.1}}.
\newblock \urlprefix\url{https://doi.org/10.5281/zenodo.1133970}.

\bibitemdeclare{misc}{matitaarith}
\bibitem{matitaarith}
\bibinfo{author}{The~Matita \surnamestart development team\surnameend}
  (\bibinfo{year}{2018}): \emph{\bibinfo{title}{Arithmetic library}}.
\newblock
  \urlprefix\url{https://github.com/LPCIC/matita/tree/master/matita/matita/lib/arithmetics}.

\bibitemdeclare{misc}{interop}
\bibitem{interop}
\bibinfo{author}{François \surnamestart Thiré\surnameend}
  (\bibinfo{year}{2018}): \emph{\bibinfo{title}{Interoperability in Dedukti:
  From the Calculus of Inductive Constructions to an extension of HOL.}}
\newblock
  \bibinfo{note}{\url{http://www.lsv.fr/~fthire/research/interop/index.php}}.

\bibitemdeclare{misc}{sttforall}
\bibitem{sttforall}
\bibinfo{author}{François \surnamestart Thiré\surnameend}
  (\bibinfo{year}{2018}): \emph{\bibinfo{title}{Sharing a library between proof
  assistants: reaching out to the HOL family.}}
\newblock
  \bibinfo{note}{\url{http://www.lsv.fr/~fthire/research/sttforall/index.php}}.

\end{thebibliography}
